\documentclass[letterpaper]{article}
\usepackage{uai2019}
\usepackage[margin=1in]{geometry}

\usepackage{times}

\usepackage[round]{natbib} 
\usepackage{xcolor} 
\usepackage{amsmath} 
\usepackage{amsthm} 
\usepackage{amsfonts} 
\usepackage{nicefrac} 
\usepackage{subfigure} 
\usepackage{enumitem} 
\usepackage{setspace} 

\usepackage{tikz} 
\usetikzlibrary{arrows.meta} 
\usetikzlibrary{decorations.markings}
\usetikzlibrary{backgrounds}
\usetikzlibrary{positioning,chains,fit,shapes,calc}

\newtheorem{theorem}{Theorem}[section] 
\newtheorem{lemma}{Lemma}[section]

\newtheorem{corollary}{Corollary}[section]
\newtheorem{definition}{Definition}[section] 
\newtheorem{example}{Example}[section] 

\title{Proving the NP-completeness of optimal moral graph triangulation}

\author{ {\bf Yang Li } \\
Faculty of Information Technology \\
Monash University\\
Clayton, Australia \\
\And
{\bf Lloyd Allison}  \\
Faculty of Information Technology \\
Monash University\\
Clayton, Australia \\
\And
{\bf Kevin Korb}   \\
Faculty of Information Technology \\
Monash University\\
Clayton, Australia \\
}

\begin{document}

\maketitle

\begin{abstract}
Moral graphs were introduced in the 1980s as an intermediate step when transforming a Bayesian network to a junction tree, on which exact belief propagation can be efficiently done. The moral graph of a Bayesian network can be trivially obtained by connecting non-adjacent parents for each node in the Bayesian network and dropping the direction of each edge. Perhaps because the moralization process looks simple, there has been little attention on the properties of moral graphs and their impact in belief propagation on Bayesian networks. This paper addresses the mistaken claim that it has been previously proved that optimal moral graph triangulation with the constraints of minimum fill-in, treewidth or total states is NP-complete. The problems are in fact NP-complete, but they have not previously been proved. We now prove these. 
\end{abstract}

\section{Introduction}
One way to conduct an exact inference on a given Bayesian network (BN) is to transform it to a tree-like structure, called \textit{junction tree} (JT), and conduct inference on a junction tree instead. The process of obtaining a JT from a BN consists of moralization, triangulation and tree decomposition. \textit{Moralization} was introduced by \cite{lauritzen1988local} as connecting non-adjacent parents for each node in the BN and dropping all the directions. \textit{Tree decomposition} maps a graph $G=(V,E)$ to a tree $T$, in which each tree node $t$ is a subset $V_t$ of vertices in $V$ and satisfying the following three conditions: (1) $\cup_{t \in T} V_t = V$; (2) for each edge $e \in E$ there exists a tree node $t \in T$ s.t. $V(e) \subseteq t$; (3) if $V_{t_i} \cap V_{t_k} = I$ then $I \subseteq V_{t_j}$ for each $t_j$ that appears on the path between $t_i$ and $t_k$. 

Any graph has a tree decomposition, but not all decomposed trees are JTs. A JT is a tree decomposition s.t. each tree node is a complete subgraph. To ensure every DAG can be transformed to a JT, which represents a family of distributions that contains the distribution of the given BN, it is necessary to triangulate the DAG's moral graph. Here, \textit{triangulation} finds a set of fill-in edges, whose addition makes a graph triangulated. When the marginal distribution of individual variables is of interest, the JT algorithm sums over the other variables in a tree node. Hence, the complexity of the JT algorithm is exponential in the size of a tree node. 

Generally, a DAG may have more than one way of being triangulated. Then an optimal triangulation can be defined in terms of the  following three constraints:
\begin{itemize}
    \item minimum fill-in: deciding whether a graph can be triangulated by at most $\lambda$ fill in edges \footnote{Originally, the minimum fill-in problem for graphs was presented and proved by \cite{yannakakis1981computing} as an optimization problem. But it can be revised to a decision problem, for which the original proof also works.}; 
    \item treewidth: deciding whether a graph has treewidth at most $\omega$;  
    \item total states: deciding whether a graph can be triangulated s.t. the total number of states is at most $\delta$ when summing over all maximal cliques in the triangulated graph.
\end{itemize}
The minimum fill-in problem is appealing because the treewidth usually increases exponentially in the number of fill-in edges. The total states problem incorporates both clique size and the number of states per variable, so is also essential to the complexity of the JT algorithm. The minimum fill-in and the treewidth problems for graphs were proved to be NP-complete by \cite{yannakakis1981computing} and \cite{arnborg1987complexity}, respectively. Each proof stated a polynomial time reduction from a known NP-complete problem to optimally triangulating specially constructed graphs that are not moral. These reductions are sufficient to show the NP-completeness of the minimum fill-in and treewidth problems for graphs, but the difficulty of these problems do not automatically carry over to moral graphs. These works, however, were cited in \cite{lauritzen1988local} (Section 6 and discussion with Augustin Kong) during the discussion of triangulating moral graphs. So it gives the impression that the minimum fill-in and treewidth problems for moral graphs were proved to be NP-complete. Based on a similar reduction, \cite{wen1990optimal} presented a proof of the NP-completeness of the total states problem for moral graphs. The proof is insufficient to support his claim, for the same reason above. Since then, all three works have being inaccurately cited as proving the NP-completeness of optimally triangulating moral graphs, e.g., \cite{kjaerulff1990triangulation, larranaga1997decomposing, amir2001efficient, flores2007review, ottosen2012all, li2017extended} etc. 

This paper proves that the minimum fill-in, treewidth or total states problems for moral graphs are indeed NP-complete. It applies an additional step to each polynomial transformation to ensure the built graphs are moral after revision. Section \ref{sec:pre} introduces equivalent properties to graph morality and the necessary concepts for the proofs. Section \ref{sec:triangulation} demonstrates why the original constructions cannot produce moral graphs and presents a fix to each problem. 

\section{Preliminary}
\label{sec:pre}
Throughout this report, unless mentioned otherwise, all graphs are assumed to be simple, connected and undirected. $G=(V,E)$ is used to denote a graph, whose vertex set is $V$ and edge set is $E$. For $uv \in E$, the subtraction $G-u$ denotes the induced subgraph $G[V \setminus \{u\}]$ and $G-uv$ denotes the subgraph $(V,E\setminus \{uv\})$.

\begin{definition}
The \textbf{deficiency} of a vertex $x$ in a graph $G=(V,E)$ is $D_G(x)=\{uv \notin E \mid u, v \in N_G(x)\}$.
\end{definition}

\begin{definition}
A vertex $x$ in a graph $G$ is \textbf{simplicial} if $D_G(x)=\emptyset$. 
\end{definition}

\begin{definition}
An \textbf{ordering} of a graph $G=(V,E)$ is a bijection $\alpha: \{1, \dots, |V|\} \leftrightarrow V$. 
\end{definition}

For convenience, let $\alpha(0)=\emptyset$. Then define the subgraph $G^i=G-\{\alpha(0), \dots, \alpha(i)\}$ for $i \in [0,|V|]$. It is called the \textit{elimination graph} (w.r.t. $\alpha$) if for each $j \in [1,i]$ the node $\alpha(j)$ is simplicial in $G^{j-1}$.
\begin{definition}
The \textbf{triangulation} of a graph $G$ w.r.t. an ordering $\alpha$ is the set of edges $H_G(\alpha)=\{D_{G^{i-1}}(\alpha(i)) \mid i \in [1,|V|]\}$. 
\end{definition}
The above definition implies that in the triangulated graph $F=G+H_G(\alpha)$, the node $\alpha(i)$ is simplicial in $F^{i-1}$ for each $i \in [1,|V|]$. 

To distinguish them from undirected edges, the directed edge from $u$ to $v$ is denoted by $\overrightarrow{uv}$. The \textit{skeleton} of a (partially) directed graph is the undirected graph obtained by dropping the direction of each directed edges. In a directed graph, $u$ is a \textit{parent} of $v$, denoted by $u \in P_G(v)$, if there is a directed edge $\overrightarrow{uv}$. 

\begin{definition}
\label{def:moral}
The \textbf{moral graph} of a directed acyclic graph $G=(V,E)$ is the skeleton of the graph $(V,E\cup F)$, where $F=\{uv \mid u,v \in P_G(x) \text{ s.t. } \{\overrightarrow{uv}, \overrightarrow{vu}\}\cap E = \emptyset\}$.
\end{definition}

\begin{figure}
\centering
\subfigure[]{
\begin{tikzpicture}[scale=0.9]
\label{fg:exp_dag}
\begin{scope}[>={Stealth[black]},every edge/.style={draw=black}]
    \node (A) at (0,0) {$v_1$};
    \node (B) at (1,0) {$v_2$};
    \node (C) at (0,-1) {$v_3$};
    \node (D) at (1,-1) {$v_4$};
    \node (E) at (0.5,-2) {$v_5$};
    \path [->] (A) edge (B);
    \path [->] (A) edge (C);
    \path [->] (B) edge (D);
    \path [->] (C) edge (E);
    \path [->] (D) edge (E);
\end{scope}
\end{tikzpicture}
}\hspace{0.5cm}
\subfigure[]{
\begin{tikzpicture}[scale=0.9]
\label{fg:exp_moral}
\begin{scope}[>={Stealth[black]},every edge/.style={draw=black}]
    \node (F) at (3,0) {$v_1$};
    \node (G) at (4,0) {$v_2$};
    \node (H) at (3,-1) {$v_3$};
    \node (I) at (4,-1) {$v_4$};
    \node (J) at (3.5,-2) {$v_5$};
    \path [-] (F) edge (G);
    \path [-] (F) edge (H);
    \path [-] (G) edge (I);
    \path [-] (H) edge (I);
    \path [-] (H) edge (J);
    \path [-] (I) edge (J);
\end{scope}
\end{tikzpicture}
}\hspace{0.5cm}
\subfigure[]{
\begin{tikzpicture}[scale=0.9]
\label{fg:exp_non_moral}
\begin{scope}[>={Stealth[black]},every edge/.style={draw=black}]
    \node (F) at (3,0) {$v_1$};
    \node (G) at (4,0) {$v_2$};
    \node (H) at (3,-1) {$v_3$};
    \node (I) at (4,-1) {$v_4$};
    \node (J) at (3.5,-2) {$v_5$};
    \path [-] (F) edge (G);
    \path [-] (F) edge (H);
    \path [-] (G) edge (I);
    \path [-] (H) edge (I);
    \path [-] (H) edge (J);
\end{scope}
\end{tikzpicture}
}
\caption{(a) a DAG, (b) its moral graph, (c) a non-weakly recursively simplicial graph.}
\label{fg:envelope}
\end{figure}
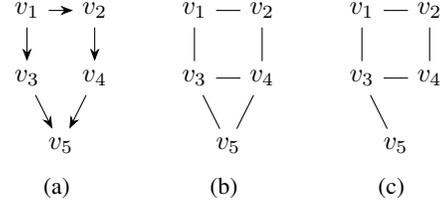

\begin{definition}
\label{def:wrs}
A graph $G=(V,E)$ is \textbf{weakly recursively simplicial} (WRS) if there exists a simplicial vertex $x \in V$ and $E'\subseteq E(G[N(x)])$ s.t. the subgraph $G'=G-x-E'$ is
weakly recursively simplicial.
\end{definition}

\begin{definition}
A set of \textbf{excesses} of a graph $G=(V,E)$ w.r.t. an ordering $\alpha$ is a bijection $\epsilon_{\alpha}: \{\alpha(1),\dots,\alpha(|V|)\} \leftrightarrow \{\epsilon_{\alpha}(\alpha(1)), \dots, \epsilon_{\alpha}(\alpha(|V|))\}$, where each excess $\epsilon_{\alpha}(\alpha(i)) \subseteq E(G[N(\alpha(i))])$ consists of some edges between the neighbours of $v_i$.
\end{definition}
The composition $\kappa=(\alpha,\epsilon_{\alpha})$ of an ordering of a graph $G$ and a set of excesses (w.r.t. $\alpha$) is called an \textit{elimination kit} of $G$. Let $\kappa(0)=\emptyset$ and $\kappa(i)=\{\alpha(i), \epsilon_{\alpha}(\alpha(i))\}$ be the $i^{th}$ elimination kit. The concept of elimination graph can be generalized to $G^i=G-\{\kappa(0),\dots,\kappa(i)\}$ for $i \in [0,|V|]$. 

\begin{definition}
Let $G=(V,E)$ be a graph and $\kappa=(\alpha, \epsilon_{\alpha})$ be an elimination kit of $G$. Then $\kappa$ is a \textbf{perfect elimination kit} (PEK) of $G$ if each node $x \in V$ satisfies $D_{G^{\alpha^{-1}(x)-1}}(x)=\emptyset$.
\end{definition}

\begin{theorem}
\label{thm:equivalent}
Let $G$ be a graph. The following are equivalent: 
\begin{enumerate}
    \item $G$ is moral.
    \item $G$ is weakly recursively simplicial.
    \item $G$ has a perfect elimination kit.
\end{enumerate}
\end{theorem}
\begin{proof}
The proof is contained in a paper that is currently under review.
\end{proof}

\begin{example}
Figure \ref{fg:exp_dag} is a DAG $G$. Figure \ref{fg:exp_moral} is the moral graph of $G$. The edge $v_3v_4$ is a filled-edge by the moralization process. The moral graph has a PEK $\kappa=(\alpha,\epsilon_{\alpha})$, where $\alpha=\{v_5,v_3,v_4,v_1,v_2\}$ and $\epsilon_{\alpha}=\{\{v_3v_4\},\emptyset,\emptyset,\emptyset,\emptyset\}$. 

Figure \ref{fg:exp_non_moral} is a non-WRS graph. $v_5$ is the only simplicial node and $E(G[N(x)])=\emptyset$. The subgraph after removing $v_5$ and the empty set of edges is a $4$-cycle that has no simplicial node.    
\end{example}

\begin{corollary}
\label{cor:chordal_implies_moral}
If a graph is chordal, then it is moral. 
\end{corollary}

Denote a bipartite graph by $G=(P \sqcup Q, E)$, where $P\sqcup Q$ represents the disjoint union of two sets in vertices of $G$.
\begin{definition}
A bipartite graph $G=(P \sqcup Q,E)$ is \textbf{chain} if there is an ordering $\alpha: \{1, \dots, |P|\} \leftrightarrow P$ s.t. the neighbours of the vertices in $P$ form a chain $N_G(\alpha(|P|)) \subseteq \dots \subseteq N_G(\alpha(1))$.
\end{definition}
The definition is also well defined for the vertices in $Q$. 

\begin{definition}
Let $G=(P\sqcup Q,E)$ be a bipartite graph. The \textbf{partition completion} of $G$ is a function $C(\cdot)$ that makes each $P$ and $Q$ a clique. 
\end{definition}
In particular, the partition completion $C_P(\cdot)$ restricted to $P$ only makes $P$ a clique.

\begin{definition}
\label{def:bip_saturated_node}
A vertex in a bipartite graph is \textbf{saturated} if it is connected to every vertex in the other partition. 
\end{definition}

\begin{example}
The bipartite graph in Figure \ref{fg:chain} with only the solid edges is not a chain graph, because the neighbour sets of $P$'s nodes do not form a chain. But the bipartite graph with all the edges is a chain graph, because $N(1)\subseteq N(2)$ w.r.t. the ordering $\alpha=\{2,1\}$. In the chain graph, the node $2$ is saturated, because it is adjacent to all nodes in $Q$. 

\begin{figure}
    \centering
    \begin{tikzpicture}[
  every node/.style={draw,circle,minimum size=1mm,inner sep=1pt},
  fsnode/.style={fill=black},
  ssnode/.style={fill=black},
  every fit/.style={ellipse,draw,inner sep=-2pt,text width=1cm, minimum height=2.5cm}
]

\node[fill=black] (f1) [label=left: $1$] at (0, 0.5) {};
\node[fill=black] (f2) [label=left: $2$] at (0, -0.5) {};

\node[fill=black] (sa) [label=right: $a$] at (4,0.6) {};
\node[fill=black] (sb) [label=right: $b$] at (4,0) {};
\node[fill=black] (sc) [label=right: $c$] at (4,-0.6) {};

\node [black,fit=(f1) (f2),label=above:$P$] at (0,0) {};
\node [black,fit=(sa) (sc),label=above:$Q$] at (4,0) {};

\path[-] (f1) edge (sa);
\path[-] (f1) edge (sb);
\path[densely dashed] (f2) edge (sa);
\path[-] (f2) edge (sb);
\path[-] (f2) edge (sc);
\end{tikzpicture}    
    \caption{A bipartite non-chain graph with only the solid edges and a bipartite chain graph with all edges. The node $2$ is a saturated node in the chain graph.}
    \label{fg:chain}
\end{figure}
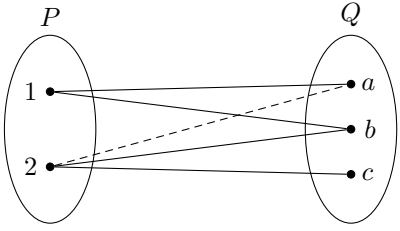

\end{example}

\section{Optimal moral graph triangulation}
\label{sec:triangulation}

\cite{yannakakis1981computing}, \cite{arnborg1987complexity} and \cite{wen1990optimal} proved a polynomial reduction from a NP-complete problem to the minimum fill-in, treewidth or total states problems for graphs, respectively. The NP-complete problems used in these proofs are the optimal linear arrangement (OLA), minimum cut linear arrangement (MCLA) and elimination degree sequence (EDS), respectively. These problems are listed below and can be found in \citep[page.~200-201]{garey2002computers}. 

\textbf{OPTIMAL LINEAR ARRANGEMENT} \\
INSTANCE: Graph $G=(V,E)$, positive integer $k \le |V|$. \\
QUESTION: Is there an ordering $\alpha:\{1, \dots, |V|\} \leftrightarrow V$ s.t. $c(\alpha)=\sum_{uv \in E}|\alpha^{-1}(u)-\alpha^{-1}(v)| \le k$? 

\textbf{MINIMUM CUT LINEAR ARRANGEMENT}\\
INSTANCE: Graph $G=(V,E)$, positive integer $k \le |V|$. \\
QUESTION: Is there an ordering $\alpha:\{1, \dots, |V|\} \leftrightarrow V$ s.t. $\forall i \in [1,|V|], |\{uv \in E \mid \alpha^{-1}(u) \le i < \alpha^{-1}(v)\}| \le k$?

\textbf{ELIMINATION DEGREE SEQUENCE}\\
INSTANCE: Graph $G=(V,E)$, sequence $<d_1, \dots, d_{|V|}>$ of non-negative integers not exceeding $|V|-1$. \\
QUESTION: Is there an ordering $\alpha:\{1, \dots, |V|\} \leftrightarrow V$ s.t. $\forall i \in [1, |V|]$, if $\alpha^{-1}(v)=i$ then there are exactly $d_i$ vertices $u$ s.t. $\alpha^{-1}(u) > i$ and $uv \in E$?

Each of the above problems asks if there exists an ordering $\alpha$ of a given graph s.t. a certain constraint is satisfied. By fixing a node $u \in V$ at an arbitrary place in an ordering $\alpha$, each of the OLA, MCLA and EDS problems seeks for a restricted ordering from the subset $A=\{\alpha \mid \alpha(i)=u\}$ of orderings of $G$ to satisfy its constraint. It is easy to verify that these restricted problems remain NP-complete. Because if there is an $O(|V|^k)$ time algorithm to answer the question within the restricted domain $A$, it takes $|V|\times O(|V|^k)$ time to answer the original question in the entire set of orderings. 

The motivation for creating a bipartite graph is the relation between chain graphs and chordal graphs stated next. 
\begin{lemma}
\label{lm:yannakakis_chain_chordal}
\citep{yannakakis1981computing}. $C(G')$ is chordal if and only if $G'$ is a bipartite chain graph. 
\end{lemma}

It follows from the lemma that triangulation of the graph $C(G')$ is equivalent to making $G'$ a chain graph. The definition of WRS implies that having at least one simplicial node is necessary for graph morality. The next lemma proves a necessary condition for the partition completion of a bipartite graph to have at least one simplicial node. 
\begin{lemma}
\label{lm:nec_cond}
If a bipartite graph $G'=(P\sqcup Q,E')$ has no saturated node, then $C(G')$ has no simplcial node. 
\end{lemma}
\begin{proof}
Assume without loss of generality that $C(G')$ has a simplicial node $x \in Q$. The graph $G'$ being connected implies $x$ has non-empty neighbours $N_{G'}(x) \subseteq P$. In order for $x$ to be simplicial in $C(G')$, the two subsets of nodes $N_{G'}(x)$ and $Q$ must form a clique. It then follows that each node $y \in N_{G'}(x) \subseteq P$ must adjacent to all nodes in $Q$. Hence, $G'$ has at least one saturated node located in $P$.  
\end{proof}

\begin{lemma}
\label{lm:cp_bip_chordal}
If $G'=(P\sqcup Q,E')$ is a bipartite graph, then $C_p(G')$ is triangulated. 
\end{lemma}
\begin{proof}
Trivial. 
\end{proof}

\subsection{Minimum fill-in}
\label{subsec:mini_fillin}

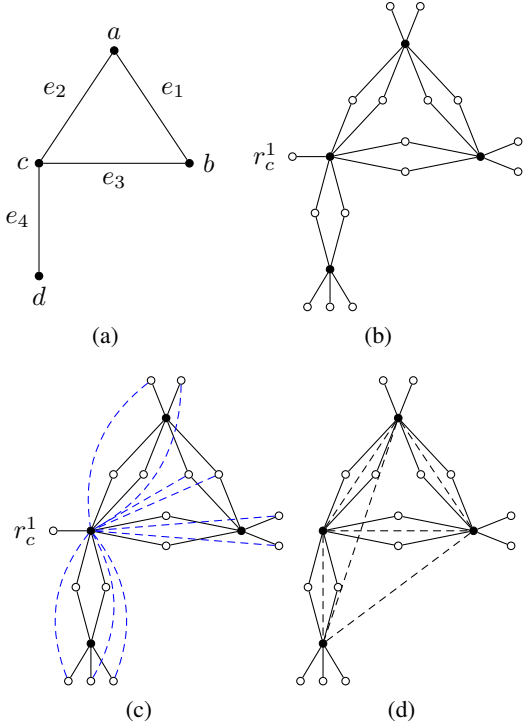
\begin{figure}
\centering
\subfigure[]{
\label{fg:g}
\begin{tikzpicture}[scale=1]
\begin{scope}[every node/.style={circle,draw,fill=black,minimum size=1mm,inner sep=1pt},>={Stealth[black]}]
    \node[label=above:$a$] (A) at (0,0.5) {};
    \node[label=left:$c$] (C) at (-1,-1) {};
    \node[label=right:$b$] (B) at (1,-1) {};
    \node[label=below:$d$] (D) at (-1,-2.5) {};
\end{scope}

\begin{scope}[every edge/.style={draw=black}]
    \path [-] (A) edge node[above right]{$e_1$} (B);
    \path [-] (A) edge node[above left]{$e_2$} (C);
    \path [-] (B) edge node[below]{$e_3$} (C);
    \path [-] (C) edge node[left]{$e_4$} (D);
\end{scope}
\end{tikzpicture}
}
\subfigure[]{
\label{fg:g'_yannakakis}
\begin{tikzpicture}[scale=1]
\begin{scope}[every node/.style={circle,draw,fill=black,minimum size=1mm,inner sep=1pt},>={Stealth[black]}]
    \node (A) at (0,0.5) {};
    \node[fill=white] (a1) at (-0.2,1) {};
    \node[fill=white] (a2) at (0.2,1) {};    
    \node (C) at (-1,-1) {};
    \node[fill=white,label=left:$r_c^1$] (c1) at (-1.5,-1) {};    
    \node (B) at (1,-1) {};
    \node[fill=white] (b1) at (1.5,-1.2) {};
    \node[fill=white] (b2) at (1.5,-0.8) {};
    \node (D) at (-1,-2.5) {};
    \node[fill=white] (d1) at (-1.3,-3) {};    
    \node[fill=white] (d2) at (-1,-3) {};    
    \node[fill=white] (d3) at (-0.7,-3) {};            
    \node[fill=white] (E11) at (0.7,-0.25) {};
    \node[fill=white] (E12) at (0.3,-0.25) {};
    \node[fill=white] (E21) at (-0.7,-0.25) {};
    \node[fill=white] (E22) at (-0.3,-0.25) {};
    \node[fill=white] (E31) at (0,-1.2) {};
    \node[fill=white] (E32) at (0,-0.8) {};
    \node[fill=white] (E41) at (-1.2,-1.75) {};
    \node[fill=white] (E42) at (-0.8,-1.75) {};
    \path [-] (A) edge (a1);
    \path [-] (A) edge (a2);  
    \path [-] (B) edge (b1);
    \path [-] (B) edge (b2);  
    \path [-] (C) edge (c1);  
    \path [-] (D) edge (d1); 
    \path [-] (D) edge (d2); 
    \path [-] (D) edge (d3);   
    \path [-] (A) edge (E11);
    \path [-] (A) edge (E12);  
    \path [-] (B) edge (E11);
    \path [-] (B) edge (E12);   
    \path [-] (A) edge (E21);
    \path [-] (A) edge (E22);  
    \path [-] (C) edge (E21);
    \path [-] (C) edge (E22); 
    \path [-] (B) edge (E31);
    \path [-] (B) edge (E32);  
    \path [-] (C) edge (E31);
    \path [-] (C) edge (E32); 
    \path [-] (D) edge (E41);
    \path [-] (D) edge (E42);  
    \path [-] (C) edge (E41);
    \path [-] (C) edge (E42);
\end{scope}
\end{tikzpicture}
}

\subfigure[]{
\label{fg:g_hat_yannakakis}
\begin{tikzpicture}[scale=1]
\begin{scope}[every node/.style={circle,draw,fill=black,minimum size=1mm,inner sep=1pt},>={Stealth[black]}]
    \node (A) at (0,0.5) {};
    \node[fill=white] (a1) at (-0.2,1) {};
    \node[fill=white] (a2) at (0.2,1) {};    
    \node (C) at (-1,-1) {};
    \node[fill=white,label=left:$r_c^1$] (c1) at (-1.5,-1) {};    
    \node (B) at (1,-1) {};
    \node[fill=white] (b1) at (1.5,-1.2) {};
    \node[fill=white] (b2) at (1.5,-0.8) {};
    \node (D) at (-1,-2.5) {};
    \node[fill=white] (d1) at (-1.3,-3) {};    
    \node[fill=white] (d2) at (-1,-3) {};    
    \node[fill=white] (d3) at (-0.7,-3) {};            
    \node[fill=white] (E11) at (0.7,-0.25) {};
    \node[fill=white] (E12) at (0.3,-0.25) {};
    \node[fill=white] (E21) at (-0.7,-0.25) {};
    \node[fill=white] (E22) at (-0.3,-0.25) {};
    \node[fill=white] (E31) at (0,-1.2) {};
    \node[fill=white] (E32) at (0,-0.8) {};
    \node[fill=white] (E41) at (-1.2,-1.75) {};
    \node[fill=white] (E42) at (-0.8,-1.75) {};
    \path [-] (A) edge (a1);
    \path [-] (A) edge (a2);  
    \path [-] (B) edge (b1);
    \path [-] (B) edge (b2);  
    \path [-] (C) edge (c1);  
    \path [-] (D) edge (d1); 
    \path [-] (D) edge (d2); 
    \path [-] (D) edge (d3);   
    \path [-] (A) edge (E11);
    \path [-] (A) edge (E12);  
    \path [-] (B) edge (E11);
    \path [-] (B) edge (E12);   
    \path [-] (A) edge (E21);
    \path [-] (A) edge (E22);  
    \path [-] (C) edge (E21);
    \path [-] (C) edge (E22); 
    \path [-] (B) edge (E31);
    \path [-] (B) edge (E32);  
    \path [-] (C) edge (E31);
    \path [-] (C) edge (E32); 
    \path [-] (D) edge (E41);
    \path [-] (D) edge (E42);  
    \path [-] (C) edge (E41);
    \path [-] (C) edge (E42);
    
    \path [densely dashed,bend right, blue] (a1) edge (C);
    \path [densely dashed,bend left, blue] (a2) edge (C);   
    \path [densely dashed, blue] (b1) edge (C);
    \path [densely dashed, blue] (b2) edge (C); 
    \path [densely dashed,bend left, blue] (d1) edge (C);
    \path [densely dashed,bend right, blue] (d2) edge (C);
    \path [densely dashed,bend right, blue] (d3) edge (C);
    \path [densely dashed, blue] (E11) edge (C);
    \path [densely dashed, blue] (E12) edge (C);
\end{scope}
\end{tikzpicture}
}
\subfigure[]{
\label{fg:g_hat_wrs_sub_yannakakis}
\begin{tikzpicture}[scale=1]
\begin{scope}[every node/.style={circle,draw,fill=black,minimum size=1mm,inner sep=1pt},>={Stealth[black]}]
    \node (A) at (0,0.5) {};
    \node[fill=white] (a1) at (-0.2,1) {};
    \node[fill=white] (a2) at (0.2,1) {};    
    \node (C) at (-1,-1) {};    
    \node (B) at (1,-1) {};
    \node[fill=white] (b1) at (1.5,-1.2) {};
    \node[fill=white] (b2) at (1.5,-0.8) {};
    \node (D) at (-1,-2.5) {};
    \node[fill=white] (d1) at (-1.3,-3) {};    
    \node[fill=white] (d2) at (-1,-3) {};    
    \node[fill=white] (d3) at (-0.7,-3) {};            
    \node[fill=white] (E11) at (0.7,-0.25) {};
    \node[fill=white] (E12) at (0.3,-0.25) {};
    \node[fill=white] (E21) at (-0.7,-0.25) {};
    \node[fill=white] (E22) at (-0.3,-0.25) {};
    \node[fill=white] (E31) at (0,-1.2) {};
    \node[fill=white] (E32) at (0,-0.8) {};
    \node[fill=white] (E41) at (-1.2,-1.75) {};
    \node[fill=white] (E42) at (-0.8,-1.75) {};
    \path [densely dashed] (A) edge (B);
    \path [densely dashed] (A) edge (C);
    \path [densely dashed] (A) edge (D);    
    \path [densely dashed] (C) edge (B);
    \path [densely dashed] (D) edge (B);    
    \path [densely dashed] (C) edge (D);
    \path [-] (A) edge (a1);
    \path [-] (A) edge (a2);  
    \path [-] (B) edge (b1);
    \path [-] (B) edge (b2);    
    \path [-] (D) edge (d1); 
    \path [-] (D) edge (d2); 
    \path [-] (D) edge (d3);   
    \path [-] (A) edge (E11);
    \path [-] (A) edge (E12);  
    \path [-] (B) edge (E11);
    \path [-] (B) edge (E12);   
    \path [-] (A) edge (E21);
    \path [-] (A) edge (E22);  
    \path [-] (C) edge (E21);
    \path [-] (C) edge (E22); 
    \path [-] (B) edge (E31);
    \path [-] (B) edge (E32);  
    \path [-] (C) edge (E31);
    \path [-] (C) edge (E32); 
    \path [-] (D) edge (E41);
    \path [-] (D) edge (E42);  
    \path [-] (C) edge (E41);
    \path [-] (C) edge (E42);
\end{scope}
\end{tikzpicture}
}
\caption{\subref{fg:g} a graph $G$; \subref{fg:g'_yannakakis} the bipartite graph $G'$ transformed from $G$ by steps Y1-Y3; \subref{fg:g_hat_yannakakis} the bipartite graph $\hat{G}$ transformed from $G'$ by saturating the node $c$ by the additional step L4; \subref{fg:g_hat_wrs_sub_yannakakis} the subgraph obtained from $C(\hat{G})$ by removing a simplicial node $r_c^1$ and its excess (specified in Lemma \ref{lm:make_cg'_moral}).}
\label{fg:yannakakis_reduction}
\end{figure}

\cite{yannakakis1981computing} presented a polynomial transformation from an instance of the OLA problem into an instance of the minimum fill-in problem for graphs. The process first takes a graph $G=(V,E)$ (Figure \ref{fg:g}) and transforms it into a bipartite graph $G'=(P\sqcup Q,E')$ (Figure \ref{fg:g'_yannakakis}) by the following steps:
\begin{enumerate}[label=Y\arabic*.]
    \item $P = \{u \mid u \in V\}$,
    \item $Q=\{e_i^j \mid j \in \{1,2\}, e_i \in E\} \cup \{R(u) \mid u \in V\}$, where $R(u) = \{r_u^j \mid j \in \{1, \dots, |V|-d_G(u)\}\}$, 
    \item $E'=\{ue_i^j \mid j \in \{1,2\}, e_i \in E \text{ s.t. } u \in V(e_i)\}\cup \{uv \mid v \in R(u), u \in P\}$.
\end{enumerate}
It then applies the partition completion on the bipartite graph $G'$ to obtain the graph $C(G')$, on which the minimum fill-in problem is solved. As can be seen, each edge node $e_i^j \in Q$ is incident to exactly two nodes in $P$, so $G'$ has no saturated node, unless $G=uv$. This implies by Lemma \ref{lm:nec_cond} that $C(G')$ is not moral. The key to make $C(G')$ moral is the additional step
\begin{enumerate}[label=L\arabic*.]
     \setcounter{enumi}{3}
    \item for a given node $u \in V$, let $\hat{E} = E' \cup S(u)$, where $S(u)=\{uv \notin E' \mid v \in Q\}$
\end{enumerate}
that makes a given node $u$ saturated in $\hat{G}$ (Figure \ref{fg:g_hat_yannakakis}). It is easy to see that the modified transformation can still be done in polynomial time, because the number of edges added by L4 is linear to the number of nodes in $Q$. 

Having a simplicial node is a necessary but not sufficient condition for being moral. The next lemma proves why adding L4 to Y1-Y3 guarantees the morality of $C(\hat{G})$, so that it becomes an instance of the minimum fill-in problem for moral graphs. 

\begin{lemma}
\label{lm:make_cg'_moral}
Let $\hat{G}=(P\sqcup Q,\hat{E})$ be the bipartite graph constructed from a graph $G=(V,E)$ by steps Y1-Y3 \& L4 for a given node $w \in V$. Then $C(\hat{G})$ is moral. 
\end{lemma}
\begin{proof}
Since $w$ is a saturated node and the partition completion makes $P$ and $Q$ cliques, the neighbour set $N_{C(\hat{G})}(r_w^1)=\{w,Q-r_w^1\}$ forms a clique. Hence $r_w^1$ is a simplicial node in $C(\hat{G})$. Removing $r_w^1$ and its excess $\epsilon(r_w^1)=S(w) \cup \{uv \mid u,v \in Q \text{ s.t. } u,v \neq r_w^1\}$, the resulting subgraph (Figure \ref{fg:g_hat_wrs_sub_yannakakis}) is the same as $C_p(G'-r_w^1)$. By Lemma \ref{lm:cp_bip_chordal}, Corollary \ref{cor:chordal_implies_moral} and Theorem \ref{thm:equivalent}, $C_p(G'-r_w^1)$ is triangulated and so WRS. It follows that $C(\hat{G})$ is also WRS and thus moral by Theorem \ref{thm:equivalent}. 
\end{proof}

It remains to show that a \textit{Yes} instance of the restricted OLA problem is also a \textit{Yes} instance of the minimum fill-in problem for moral graphs and vice versa. To do so, the next lemma calculates the difference between the cost of a graph $G$ w.r.t. an ordering $\alpha$ and the number of fill-in edges that triangulates the corresponding moral graph $C(\hat{G})$ w.r.t. $\alpha$. And it proves that the difference is a constant for any restricted ordering $\alpha$. Define the cost of an edge $e=uv \in E$ w.r.t. an ordering $\alpha$ to be $\delta(e,\alpha)=|\alpha^{-1}(u)-\alpha^{-1}(v)|$. 

\begin{lemma}
\label{lm:mini_chain_fillin_npc}
Given a graph $G=(V,E)$ and a positive integer $k \le |V|$, for any node $w \in V$ the minimum cost w.r.t. an ordering $\alpha$ of $G$ is $k$ with $\alpha^{-1}(w)=|V|$ if and only if the corresponding moral graph $C(\hat{G})$ can be triangulated by  $\lambda=k+\frac{|V|(|V|-1)(|V|-2)}{2} - 2|E| + d_G(w)$ fill in edges w.r.t. $\alpha$.  
\end{lemma}
\begin{proof}
Let $\hat{G}$ be the polynomial transformed graph from $G=(V,E)$ using steps Y1-Y3 $\&$ L4. For a given node $w \in V$, an ordering $\alpha \in A=\{\alpha \mid \alpha(1)=w\}$ uniquely specifies a set $F_{\hat{G}}(\alpha)$ of fill-in edges to make $\hat{G}$ a chain by the following two steps:
\begin{enumerate}[label=\alph*.]
    \item for each node $u \in Q$ calculate $\sigma(u)=\max\{i\mid u\alpha(i) \in E\}$,
    \item for any ordering $\alpha$, define $F_{\hat{G}}(\alpha)=\{u\alpha(j) \notin \hat{E} \mid c\neq u \in Q, j < \sigma(u)\}$.
\end{enumerate}
It is easy to see that $F_{\hat{G}}(\alpha)$ is minimal because any edge deletion from it stops the neighbours of $P$'s nodes in $\hat{G}+F_{\hat{G}}(\alpha)$ from forming a chain. Lemma \ref{lm:yannakakis_chain_chordal} implies $C(\hat{G})+F_{\hat{G}}(\alpha)$ is triangulated, so $F_{\hat{G}}(\alpha)$ is a minimal triangulation of $C(\hat{G})$. It remains to show that for every ordering $\alpha \in A$, the cardinality of $F_{\hat{G}}(\alpha)$ satisfies 
\begin{equation}
\label{eq:h_alpha}
    f_{\hat{G}}(\alpha)=c(\alpha)+\frac{|V|(|V|-1)(|V|-2)}{2} - 2|E| + d_G(w),
\end{equation}
where $c(\alpha)$ is the total cost of $G$ w.r.t. $\alpha$. \cite{yannakakis1981computing} proved that for every ordering $\pi$ (not necessarily in $A$), the number of fill-in edges
\begin{equation}
\label{eq:yannakakis'_h_alpha}
f_{G'}(\pi)=c(\pi) + \frac{|V|^2(|V|-1)}{2}-2|E|.
\end{equation}
The following is a brief explanation of \cite{yannakakis1981computing}'s proof of equation (\ref{eq:yannakakis'_h_alpha}). For every $v \in V$, each $x \in R(v)$ connects to $\pi^{-1}(v)-1$ nodes in $P$, whose orderings are smaller than $\sigma(x)$. For any edge $e=uv \in E$, assume without loss of generality that  $\pi^{-1}(u) < \pi^{-1}(v)$. Since each $e^j$ in $\hat{G}$ is adjacent to the two end nodes of the edge $e$, step (b) connects $e^j$ to $\pi^{-1}(v)-2=\pi^{-1}(u)+[\pi^{-1}(v)-\pi^{-1}(u)]-2=\pi^{-1}(u)+\delta(e,\pi)-2$ nodes in $P$. Hence, $F_{G'}(\pi)$ contains $\pi^{-1}(v)+\pi^{-1}(u)+\delta(e,\pi)-4$ edges incident to both $e^1$ and $e^2$. Therefore, the number of fill in edges w.r.t to $\pi$ is calculated by 
\begin{align*}
f_{G'}(\pi)= &\sum_{v \in V} \sum_{x \in R(v)} [\pi^{-1}(v)-1] +\\
& \sum_{e=uv\in E} [\pi^{-1}(v)+\pi^{-1}(u)+\delta(e,\pi)-4] \\
= &\sum_{v \in V} [|V|-d_G(v)][\pi^{-1}(v)-1] +\\
& \sum_{v\in V} d_G(v)\pi^{-1}(v) + \sum_{e\in E} \delta(e, \pi) - 4|E|\\
=& \sum_{v \in V} |V|[\pi^{-1}(v)-1] +\\
& \sum_{v\in V} d_G(v) + c(\pi) - 4|E|\\
=& c(\pi) + \frac{|V|^2(|V|-1)}{2}-2|E|,
\end{align*}
where by definition $\sum_{e \in E} \delta(e,\pi) = c(\pi)$ and the last equality obtains because $\sum_{v\in V} d_G(v)=2|E|$ and $\sum_{v \in V} (\pi^{-1}(v)-1)=|V|(|V|-1)/2$. Note that the reason to make two edge nodes and the residual nodes is to cancel the term $d_G(v)$ during the derivation of $f_{G'}(\pi)$, so that the difference between $f_{G'}(\pi)$ and $c(\pi)$ is constant, regardless of the corresponding ordering. 

The size difference between the two sets of edges $\hat{E}$ and $E'$ is 
\begin{align}
\label{eq:sc}
|S(w)|=& 2(|E|-d_G(w))+\sum_{u \in V}^{u \neq w} |R(u)| \nonumber \\
=& 2(|E|-d_G(w))+\sum_{u \in V}^{u \neq w}(|V|-d_G(u)) \nonumber \\
= &2(|E|-d_G(w))+|V|(|V|-1) - \nonumber \\
&\sum_{u \in V}d_G(u) + d_G(w) \nonumber \\
= & |V|(|V|-1) - d_G(w).
\end{align}
Since equation (\ref{eq:yannakakis'_h_alpha}) is true for every ordering, it certainly holds for orderings in $A$. It follows from $\alpha^{-1}(w)=1$ that $S(w) \subseteq F_{G'}(\alpha)$. Hence, equation (\ref{eq:h_alpha}) is obtained by subtracting equation (\ref{eq:sc}) from equation (\ref{eq:yannakakis'_h_alpha}). If there exists an ordering $\alpha$ of $G$, w.r.t. which the minimum cost of $G$ is $k$, then $\alpha$ produces a set of fill-in edges that triangulates the moral graph $C(\hat{G})$ with $\lambda$ edges. Conversely, if the moral graph $C(\hat{G})$ can be triangulated w.r.t. an ordering $\alpha$ with $\lambda$ fill-in edges, $\alpha$ indicates the minimum cost of the graph $G$ is $k$. 
\end{proof}

\begin{theorem}
The minimum fill-in problem for moral graphs is NP-complete.
\end{theorem}
\begin{proof}
 Since any set of fill-in edges that triangulates a moral graph can be verified in polynomial time to have at most $\lambda$ edges or not, the minimum fill-in problem for moral graphs is in NP. Given a graph $G$ can be polynomially transformed to the corresponding moral graph $C(\hat{G})$, Lemma \ref{lm:mini_chain_fillin_npc} proves the NP-hardness of the problem. 
\end{proof}

\subsection{Treewidth}

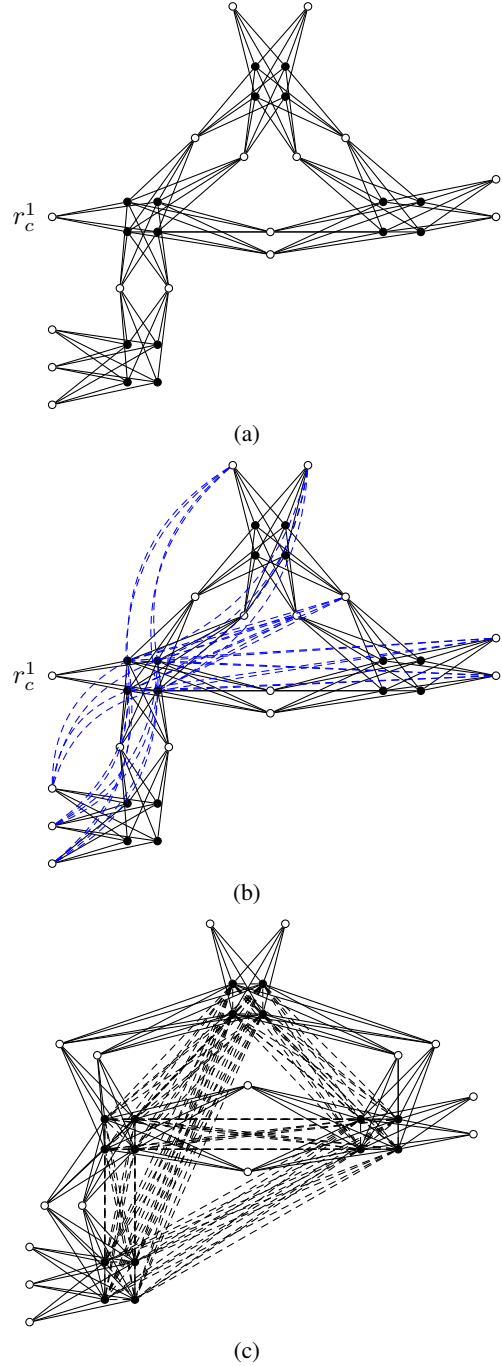
\begin{figure}
\centering
\subfigure[]{
\label{fg:g'_arnborg}
\begin{tikzpicture}[scale=1]
\begin{scope}[every node/.style={circle,draw,fill=black,minimum size=1mm,inner sep=1pt},>={Stealth[black]}]
	\node[fill=white] (E11) at (1,-0.25) {};
    \node[fill=white] (E12) at (0.35,-0.5) {};    
    \node[fill=white] (E21) at (-1,-0.25) {};
    \node[fill=white] (E22) at (-0.35,-0.5) {};    
    \node[fill=white] (E31) at (0,-1.8) {};
    \node[fill=white] (E32) at (0,-1.5) {};    
    \node[fill=white] (E41) at (-2,-2.25) {};
    \node[fill=white] (E42) at (-1.35,-2.25) {};
    
    \node (A1) at (-0.2,0.3) {};
    \node (A2) at (-0.2,0.7) {};
    \node (A3) at (0.2,0.3) {};
    \node (A4) at (0.2,0.7) {};
    \node[fill=white] (a1) at (-0.5,1.5) {};
    \node[fill=white] (a2) at (0.5,1.5) {};   
      
    \path [-] (A1) edge (a1);
    \path [-] (A1) edge (a2); 
    \path [-] (A2) edge (a1);
    \path [-] (A2) edge (a2); 
    \path [-] (A3) edge (a1);
    \path [-] (A3) edge (a2); 
    \path [-] (A4) edge (a1);
    \path [-] (A4) edge (a2); 
    
    \path [-] (A1) edge (E11);
    \path [-] (A1) edge (E12); 
    \path [-] (A2) edge (E11);
    \path [-] (A2) edge (E12);
    \path [-] (A3) edge (E11);
    \path [-] (A3) edge (E12);
    \path [-] (A4) edge (E11);
    \path [-] (A4) edge (E12);
    
    \path [-] (A1) edge (E21);
    \path [-] (A1) edge (E22); 
    \path [-] (A2) edge (E21);
    \path [-] (A2) edge (E22);
    \path [-] (A3) edge (E21);
    \path [-] (A3) edge (E22);
    \path [-] (A4) edge (E21);
    \path [-] (A4) edge (E22);
    
    \node (B1) at (1.5,-1.5) {};
    \node (B2) at (2,-1.5) {};
    \node (B3) at (1.5,-1.1) {};
    \node (B4) at (2,-1.1) {};
    \node[fill=white] (b1) at (3,-0.8) {};
    \node[fill=white] (b2) at (3,-1.3) {};
    
    \path [-] (B1) edge (b1);
    \path [-] (B1) edge (b2); 
    \path [-] (B2) edge (b1);
    \path [-] (B2) edge (b2); 
    \path [-] (B3) edge (b1);
    \path [-] (B3) edge (b2); 
    \path [-] (B4) edge (b1);
    \path [-] (B4) edge (b2);
    
    \path [-] (B1) edge (E11);
    \path [-] (B1) edge (E12); 
    \path [-] (B2) edge (E11);
    \path [-] (B2) edge (E12);
    \path [-] (B3) edge (E11);
    \path [-] (B3) edge (E12);
    \path [-] (B4) edge (E11);
    \path [-] (B4) edge (E12);
    
    \path [-] (B1) edge (E31);
    \path [-] (B1) edge (E32); 
    \path [-] (B2) edge (E31);
    \path [-] (B2) edge (E32);
    \path [-] (B3) edge (E31);
    \path [-] (B3) edge (E32);
    \path [-] (B4) edge (E31);
    \path [-] (B4) edge (E32);
    
    \node (C1) at (-1.5,-1.5) {};
    \node (C2) at (-1.9,-1.5) {};
    \node (C3) at (-1.5,-1.1) {};
    \node (C4) at (-1.9,-1.1) {};
    \node[fill=white,label=left:$r_c^1$] (c1) at (-2.9,-1.3) {};
    
    \path [-] (C1) edge (c1);
    \path [-] (C2) edge (c1);
    \path [-] (C3) edge (c1);
    \path [-] (C4) edge (c1);
       
    \path [-] (C1) edge (E21);
    \path [-] (C1) edge (E22); 
    \path [-] (C2) edge (E21);
    \path [-] (C2) edge (E22);
    \path [-] (C3) edge (E21);
    \path [-] (C3) edge (E22);
    \path [-] (C4) edge (E21);
    \path [-] (C4) edge (E22);
    
    \path [-] (C1) edge (E31);
    \path [-] (C1) edge (E32); 
    \path [-] (C2) edge (E31);
    \path [-] (C2) edge (E32);
    \path [-] (C3) edge (E31);
    \path [-] (C3) edge (E32);
    \path [-] (C4) edge (E31);
    \path [-] (C4) edge (E32);
    
    \path [-] (C1) edge (E41);
    \path [-] (C1) edge (E42); 
    \path [-] (C2) edge (E41);
    \path [-] (C2) edge (E42);
    \path [-] (C3) edge (E41);
    \path [-] (C3) edge (E42);
    \path [-] (C4) edge (E41);
    \path [-] (C4) edge (E42);
    
    \node (D1) at (-1.5,-3) {};
    \node (D2) at (-1.9,-3) {};
    \node (D3) at (-1.5,-3.5) {};
    \node (D4) at (-1.9,-3.5) {};
    \node[fill=white] (d1) at (-2.9,-2.8) {};
    \node[fill=white] (d2) at (-2.9,-3.3) {};
    \node[fill=white] (d3) at (-2.9,-3.8) {};
    
    \path [-] (D1) edge (d1);
    \path [-] (D1) edge (d2); 
    \path [-] (D1) edge (d3);
    \path [-] (D2) edge (d1);
    \path [-] (D2) edge (d2); 
    \path [-] (D2) edge (d3); 
    \path [-] (D3) edge (d1);
    \path [-] (D3) edge (d2); 
    \path [-] (D3) edge (d3);
    \path [-] (D4) edge (d1);
    \path [-] (D4) edge (d2);
    \path [-] (D4) edge (d3);
    
    \path [-] (D1) edge (E41);
    \path [-] (D1) edge (E42); 
    \path [-] (D2) edge (E41);
    \path [-] (D2) edge (E42);
    \path [-] (D3) edge (E41);
    \path [-] (D3) edge (E42);
    \path [-] (D4) edge (E41);
    \path [-] (D4) edge (E42);
\end{scope}
\end{tikzpicture}
}\hspace{1cm}
\subfigure[]{
\label{fg:g_hat_arnborg}
\begin{tikzpicture}[scale=1]
\begin{scope}[every node/.style={circle,draw,fill=black,minimum size=1mm,inner sep=1pt},>={Stealth[black]}]
    \node[fill=white] (E11) at (1,-0.25) {};
    \node[fill=white] (E12) at (0.35,-0.5) {};    
    \node[fill=white] (E21) at (-1,-0.25) {};
    \node[fill=white] (E22) at (-0.35,-0.5) {};    
    \node[fill=white] (E31) at (0,-1.8) {};
    \node[fill=white] (E32) at (0,-1.5) {};    
    \node[fill=white] (E41) at (-2,-2.25) {};
    \node[fill=white] (E42) at (-1.35,-2.25) {};
    
    \node (A1) at (-0.2,0.3) {};
    \node (A2) at (-0.2,0.7) {};
    \node (A3) at (0.2,0.3) {};
    \node (A4) at (0.2,0.7) {};
    \node[fill=white] (a1) at (-0.5,1.5) {};
    \node[fill=white] (a2) at (0.5,1.5) {};   
      
    \path [-] (A1) edge (a1);
    \path [-] (A1) edge (a2); 
    \path [-] (A2) edge (a1);
    \path [-] (A2) edge (a2); 
    \path [-] (A3) edge (a1);
    \path [-] (A3) edge (a2); 
    \path [-] (A4) edge (a1);
    \path [-] (A4) edge (a2); 
    
    \path [-] (A1) edge (E11);
    \path [-] (A1) edge (E12); 
    \path [-] (A2) edge (E11);
    \path [-] (A2) edge (E12);
    \path [-] (A3) edge (E11);
    \path [-] (A3) edge (E12);
    \path [-] (A4) edge (E11);
    \path [-] (A4) edge (E12);
    
    \path [-] (A1) edge (E21);
    \path [-] (A1) edge (E22); 
    \path [-] (A2) edge (E21);
    \path [-] (A2) edge (E22);
    \path [-] (A3) edge (E21);
    \path [-] (A3) edge (E22);
    \path [-] (A4) edge (E21);
    \path [-] (A4) edge (E22);
    
    \node (B1) at (1.5,-1.5) {};
    \node (B2) at (2,-1.5) {};
    \node (B3) at (1.5,-1.1) {};
    \node (B4) at (2,-1.1) {};
    \node[fill=white] (b1) at (3,-0.8) {};
    \node[fill=white] (b2) at (3,-1.3) {};
    
    \path [-] (B1) edge (b1);
    \path [-] (B1) edge (b2); 
    \path [-] (B2) edge (b1);
    \path [-] (B2) edge (b2); 
    \path [-] (B3) edge (b1);
    \path [-] (B3) edge (b2); 
    \path [-] (B4) edge (b1);
    \path [-] (B4) edge (b2);
    
    \path [-] (B1) edge (E11);
    \path [-] (B1) edge (E12); 
    \path [-] (B2) edge (E11);
    \path [-] (B2) edge (E12);
    \path [-] (B3) edge (E11);
    \path [-] (B3) edge (E12);
    \path [-] (B4) edge (E11);
    \path [-] (B4) edge (E12);
    
    \path [-] (B1) edge (E31);
    \path [-] (B1) edge (E32); 
    \path [-] (B2) edge (E31);
    \path [-] (B2) edge (E32);
    \path [-] (B3) edge (E31);
    \path [-] (B3) edge (E32);
    \path [-] (B4) edge (E31);
    \path [-] (B4) edge (E32);
    
    \node (C1) at (-1.5,-1.5) {};
    \node (C2) at (-1.9,-1.5) {};
    \node (C3) at (-1.5,-1.1) {};
    \node (C4) at (-1.9,-1.1) {};
    \node[fill=white,label=left:$r_c^1$] (c1) at (-2.9,-1.3) {};
    
    \path [-] (C1) edge (c1);
    \path [-] (C2) edge (c1);
    \path [-] (C3) edge (c1);
    \path [-] (C4) edge (c1);
       
    \path [-] (C1) edge (E21);
    \path [-] (C1) edge (E22); 
    \path [-] (C2) edge (E21);
    \path [-] (C2) edge (E22);
    \path [-] (C3) edge (E21);
    \path [-] (C3) edge (E22);
    \path [-] (C4) edge (E21);
    \path [-] (C4) edge (E22);
    
    \path [-] (C1) edge (E31);
    \path [-] (C1) edge (E32); 
    \path [-] (C2) edge (E31);
    \path [-] (C2) edge (E32);
    \path [-] (C3) edge (E31);
    \path [-] (C3) edge (E32);
    \path [-] (C4) edge (E31);
    \path [-] (C4) edge (E32);
    
    \path [-] (C1) edge (E41);
    \path [-] (C1) edge (E42); 
    \path [-] (C2) edge (E41);
    \path [-] (C2) edge (E42);
    \path [-] (C3) edge (E41);
    \path [-] (C3) edge (E42);
    \path [-] (C4) edge (E41);
    \path [-] (C4) edge (E42);
    
    \node (D1) at (-1.5,-3) {};
    \node (D2) at (-1.9,-3) {};
    \node (D3) at (-1.5,-3.5) {};
    \node (D4) at (-1.9,-3.5) {};
    \node[fill=white] (d1) at (-2.9,-2.8) {};
    \node[fill=white] (d2) at (-2.9,-3.3) {};
    \node[fill=white] (d3) at (-2.9,-3.8) {};
    
    \path [-] (D1) edge (d1);
    \path [-] (D1) edge (d2); 
    \path [-] (D1) edge (d3);
    \path [-] (D2) edge (d1);
    \path [-] (D2) edge (d2); 
    \path [-] (D2) edge (d3); 
    \path [-] (D3) edge (d1);
    \path [-] (D3) edge (d2); 
    \path [-] (D3) edge (d3);
    \path [-] (D4) edge (d1);
    \path [-] (D4) edge (d2);
    \path [-] (D4) edge (d3);
    
    \path [-] (D1) edge (E41);
    \path [-] (D1) edge (E42); 
    \path [-] (D2) edge (E41);
    \path [-] (D2) edge (E42);
    \path [-] (D3) edge (E41);
    \path [-] (D3) edge (E42);
    \path [-] (D4) edge (E41);
    \path [-] (D4) edge (E42);

    \path [densely dashed,bend right,blue] (a1) edge (C1);
    \path [densely dashed,bend left,blue] (a2) edge (C1);   
    \path [densely dashed,blue] (b1) edge (C1);
    \path [densely dashed,blue] (b2) edge (C1); 
    \path [densely dashed,bend left,blue] (d1) edge (C1);
    \path [densely dashed,bend right,blue] (d2) edge (C1);
    \path [densely dashed,bend right,blue] (d3) edge (C1);
    \path [densely dashed,blue] (E11) edge (C1);
    \path [densely dashed,blue] (E12) edge (C1);
    
    \path [densely dashed,bend right,blue] (a1) edge (C2);
    \path [densely dashed,bend left,blue] (a2) edge (C2);   
    \path [densely dashed,blue] (b1) edge (C2);
    \path [densely dashed,blue] (b2) edge (C2); 
    \path [densely dashed,bend left,blue] (d1) edge (C2);
    \path [densely dashed,bend right,blue] (d2) edge (C2);
    \path [densely dashed,bend right,blue] (d3) edge (C2);
    \path [densely dashed,blue] (E11) edge (C2);
    \path [densely dashed,blue] (E12) edge (C2);
    
    \path [densely dashed,bend right,blue] (a1) edge (C3);
    \path [densely dashed,bend left,blue] (a2) edge (C3);   
    \path [densely dashed,blue] (b1) edge (C3);
    \path [densely dashed,blue] (b2) edge (C3); 
    \path [densely dashed,bend left,blue] (d1) edge (C3);
    \path [densely dashed,bend right,blue] (d2) edge (C3);
    \path [densely dashed,bend right,blue] (d3) edge (C3);
    \path [densely dashed,blue] (E11) edge (C3);
    \path [densely dashed,blue] (E12) edge (C3);
    
    \path [densely dashed,bend right,blue] (a1) edge (C4);
    \path [densely dashed,bend left,blue] (a2) edge (C4);   
    \path [densely dashed,blue] (b1) edge (C4);
    \path [densely dashed,blue] (b2) edge (C4); 
    \path [densely dashed,bend left,blue] (d1) edge (C4);
    \path [densely dashed,bend right,blue] (d2) edge (C4);
    \path [densely dashed,bend right,blue] (d3) edge (C4);
    \path [densely dashed,blue] (E11) edge (C4);
    \path [densely dashed,blue] (E12) edge (C4);
\end{scope}
\end{tikzpicture}
}\hspace{1cm}
\subfigure[]{
\label{fg:g_hat_wrs_sub_arnborg}
\begin{tikzpicture}[scale=1]
\begin{scope}[every node/.style={circle,draw,fill=black,minimum size=1mm,inner sep=1pt},>={Stealth[black]}]
    \node[fill=white] (E11) at (2,-0.25) {};
    \node[fill=white] (E12) at (2.5,-0.1) {};    
    \node[fill=white] (E21) at (-2,-0.25) {};
    \node[fill=white] (E22) at (-2.5,-0.1) {};    
    \node[fill=white] (E31) at (0,-1.8) {};
    \node[fill=white] (E32) at (0,-0.65) {};    
    \node[fill=white] (E41) at (-2.7,-2.25) {};
    \node[fill=white] (E42) at (-2.2,-2.25) {};
    
    \node (A1) at (-0.2,0.3) {};
    \node (A2) at (-0.2,0.7) {};
    \node (A3) at (0.2,0.3) {};
    \node (A4) at (0.2,0.7) {};
    \node[fill=white] (a1) at (-0.5,1.5) {};
    \node[fill=white] (a2) at (0.5,1.5) {};   
      
    \path [-] (A1) edge (a1);
    \path [-] (A1) edge (a2); 
    \path [-] (A2) edge (a1);
    \path [-] (A2) edge (a2); 
    \path [-] (A3) edge (a1);
    \path [-] (A3) edge (a2); 
    \path [-] (A4) edge (a1);
    \path [-] (A4) edge (a2); 
    
    \path [-] (A1) edge (E11);
    \path [-] (A1) edge (E12); 
    \path [-] (A2) edge (E11);
    \path [-] (A2) edge (E12);
    \path [-] (A3) edge (E11);
    \path [-] (A3) edge (E12);
    \path [-] (A4) edge (E11);
    \path [-] (A4) edge (E12);
    
    \path [-] (A1) edge (E21);
    \path [-] (A1) edge (E22); 
    \path [-] (A2) edge (E21);
    \path [-] (A2) edge (E22);
    \path [-] (A3) edge (E21);
    \path [-] (A3) edge (E22);
    \path [-] (A4) edge (E21);
    \path [-] (A4) edge (E22);
    
    \node (B1) at (1.5,-1.5) {};
    \node (B2) at (2,-1.5) {};
    \node (B3) at (1.5,-1.1) {};
    \node (B4) at (2,-1.1) {};
    \node[fill=white] (b1) at (3,-0.8) {};
    \node[fill=white] (b2) at (3,-1.3) {};
    
    \path [-] (B1) edge (b1);
    \path [-] (B1) edge (b2); 
    \path [-] (B2) edge (b1);
    \path [-] (B2) edge (b2); 
    \path [-] (B3) edge (b1);
    \path [-] (B3) edge (b2); 
    \path [-] (B4) edge (b1);
    \path [-] (B4) edge (b2);
    
    \path [-] (B1) edge (E11);
    \path [-] (B1) edge (E12); 
    \path [-] (B2) edge (E11);
    \path [-] (B2) edge (E12);
    \path [-] (B3) edge (E11);
    \path [-] (B3) edge (E12);
    \path [-] (B4) edge (E11);
    \path [-] (B4) edge (E12);
    
    \path [-] (B1) edge (E31);
    \path [-] (B1) edge (E32); 
    \path [-] (B2) edge (E31);
    \path [-] (B2) edge (E32);
    \path [-] (B3) edge (E31);
    \path [-] (B3) edge (E32);
    \path [-] (B4) edge (E31);
    \path [-] (B4) edge (E32);
    
    \node (C1) at (-1.5,-1.5) {};
    \node (C2) at (-1.9,-1.5) {};
    \node (C3) at (-1.5,-1.1) {};
    \node (C4) at (-1.9,-1.1) {};
           
    \path [-] (C1) edge (E21);
    \path [-] (C1) edge (E22); 
    \path [-] (C2) edge (E21);
    \path [-] (C2) edge (E22);
    \path [-] (C3) edge (E21);
    \path [-] (C3) edge (E22);
    \path [-] (C4) edge (E21);
    \path [-] (C4) edge (E22);
    
    \path [-] (C1) edge (E31);
    \path [-] (C1) edge (E32); 
    \path [-] (C2) edge (E31);
    \path [-] (C2) edge (E32);
    \path [-] (C3) edge (E31);
    \path [-] (C3) edge (E32);
    \path [-] (C4) edge (E31);
    \path [-] (C4) edge (E32);
    
    \path [-] (C1) edge (E41);
    \path [-] (C1) edge (E42); 
    \path [-] (C2) edge (E41);
    \path [-] (C2) edge (E42);
    \path [-] (C3) edge (E41);
    \path [-] (C3) edge (E42);
    \path [-] (C4) edge (E41);
    \path [-] (C4) edge (E42);
    
    \node (D1) at (-1.5,-3) {};
    \node (D2) at (-1.9,-3) {};
    \node (D3) at (-1.5,-3.5) {};
    \node (D4) at (-1.9,-3.5) {};
    \node[fill=white] (d1) at (-2.9,-2.8) {};
    \node[fill=white] (d2) at (-2.9,-3.3) {};
    \node[fill=white] (d3) at (-2.9,-3.8) {};
    
    \path [-] (D1) edge (d1);
    \path [-] (D1) edge (d2); 
    \path [-] (D1) edge (d3);
    \path [-] (D2) edge (d1);
    \path [-] (D2) edge (d2); 
    \path [-] (D2) edge (d3); 
    \path [-] (D3) edge (d1);
    \path [-] (D3) edge (d2); 
    \path [-] (D3) edge (d3);
    \path [-] (D4) edge (d1);
    \path [-] (D4) edge (d2);
    \path [-] (D4) edge (d3);
    
    \path [-] (D1) edge (E41);
    \path [-] (D1) edge (E42); 
    \path [-] (D2) edge (E41);
    \path [-] (D2) edge (E42);
    \path [-] (D3) edge (E41);
    \path [-] (D3) edge (E42);
    \path [-] (D4) edge (E41);
    \path [-] (D4) edge (E42);
    
    \path [densely dashed] (A1) edge (A2);
    \path [densely dashed] (A1) edge (A3);
    \path [densely dashed] (A1) edge (A4);
    \path [densely dashed] (A2) edge (A3);
    \path [densely dashed] (A2) edge (A4);
    \path [densely dashed] (A3) edge (A4);
    
    \path [densely dashed] (B1) edge (B2);
    \path [densely dashed] (B1) edge (B3);
    \path [densely dashed] (B1) edge (B4);
    \path [densely dashed] (B2) edge (B3);
    \path [densely dashed] (B2) edge (B4);
    \path [densely dashed] (B3) edge (B4);
    
    \path [densely dashed] (C1) edge (C2);
    \path [densely dashed] (C1) edge (C3);
    \path [densely dashed] (C1) edge (C4);
    \path [densely dashed] (C2) edge (C3);
    \path [densely dashed] (C2) edge (C4);
    \path [densely dashed] (C3) edge (C4);
    
    \path [densely dashed] (D1) edge (D2);
    \path [densely dashed] (D1) edge (D3);
    \path [densely dashed] (D1) edge (D4);
    \path [densely dashed] (D2) edge (D3);
    \path [densely dashed] (D2) edge (D4);
    \path [densely dashed] (D3) edge (D4);
    
    \path [densely dashed] (A1) edge (B1);
    \path [densely dashed] (A1) edge (B2);
    \path [densely dashed] (A1) edge (B3);
    \path [densely dashed] (A1) edge (B4);
    \path [densely dashed] (A2) edge (B1);
    \path [densely dashed] (A2) edge (B2);
    \path [densely dashed] (A2) edge (B3);
    \path [densely dashed] (A2) edge (B4);
    \path [densely dashed] (A3) edge (B1);
    \path [densely dashed] (A3) edge (B2);
    \path [densely dashed] (A3) edge (B3);
    \path [densely dashed] (A3) edge (B4);
    \path [densely dashed] (A4) edge (B1);
    \path [densely dashed] (A4) edge (B2);
    \path [densely dashed] (A4) edge (B3);
    \path [densely dashed] (A4) edge (B4);
    
    \path [densely dashed] (A1) edge (C1);
    \path [densely dashed] (A1) edge (C2);
    \path [densely dashed] (A1) edge (C3);
    \path [densely dashed] (A1) edge (C4);
    \path [densely dashed] (A2) edge (C1);
    \path [densely dashed] (A2) edge (C2);
    \path [densely dashed] (A2) edge (C3);
    \path [densely dashed] (A2) edge (C4);
    \path [densely dashed] (A3) edge (C1);
    \path [densely dashed] (A3) edge (C2);
    \path [densely dashed] (A3) edge (C3);
    \path [densely dashed] (A3) edge (C4);
    \path [densely dashed] (A4) edge (C1);
    \path [densely dashed] (A4) edge (C2);
    \path [densely dashed] (A4) edge (C3);
    \path [densely dashed] (A4) edge (C4);
    
    \path [densely dashed] (A1) edge (D1);
    \path [densely dashed] (A1) edge (D2);
    \path [densely dashed] (A1) edge (D3);
    \path [densely dashed] (A1) edge (D4);
    \path [densely dashed] (A2) edge (D1);
    \path [densely dashed] (A2) edge (D2);
    \path [densely dashed] (A2) edge (D3);
    \path [densely dashed] (A2) edge (D4);
    \path [densely dashed] (A3) edge (D1);
    \path [densely dashed] (A3) edge (D2);
    \path [densely dashed] (A3) edge (D3);
    \path [densely dashed] (A3) edge (D4);
    \path [densely dashed] (A4) edge (D1);
    \path [densely dashed] (A4) edge (D2);
    \path [densely dashed] (A4) edge (D3);
    \path [densely dashed] (A4) edge (D4);
    
    \path [densely dashed] (B1) edge (C1);
    \path [densely dashed] (B1) edge (C2);
    \path [densely dashed] (B1) edge (C3);
    \path [densely dashed] (B1) edge (C4);
    \path [densely dashed] (B2) edge (C1);
    \path [densely dashed] (B2) edge (C2);
    \path [densely dashed] (B2) edge (C3);
    \path [densely dashed] (B2) edge (C4);
    \path [densely dashed] (B3) edge (C1);
    \path [densely dashed] (B3) edge (C2);
    \path [densely dashed] (B3) edge (C3);
    \path [densely dashed] (B3) edge (C4);
    \path [densely dashed] (B4) edge (C1);
    \path [densely dashed] (B4) edge (C2);
    \path [densely dashed] (B4) edge (C3);
    \path [densely dashed] (B4) edge (C4);
    
    \path [densely dashed] (B1) edge (D1);
    \path [densely dashed] (B1) edge (D2);
    \path [densely dashed] (B1) edge (D3);
    \path [densely dashed] (B1) edge (D4);
    \path [densely dashed] (B2) edge (D1);
    \path [densely dashed] (B2) edge (D2);
    \path [densely dashed] (B2) edge (D3);
    \path [densely dashed] (B2) edge (D4);
    \path [densely dashed] (B3) edge (D1);
    \path [densely dashed] (B3) edge (D2);
    \path [densely dashed] (B3) edge (D3);
    \path [densely dashed] (B3) edge (D4);
    \path [densely dashed] (B4) edge (D1);
    \path [densely dashed] (B4) edge (D2);
    \path [densely dashed] (B4) edge (D3);
    \path [densely dashed] (B4) edge (D4);
    
    \path [densely dashed] (D1) edge (C1);
    \path [densely dashed] (D1) edge (C2);
    \path [densely dashed] (D1) edge (C3);
    \path [densely dashed] (D1) edge (C4);
    \path [densely dashed] (D2) edge (C1);
    \path [densely dashed] (D2) edge (C2);
    \path [densely dashed] (D2) edge (C3);
    \path [densely dashed] (D2) edge (C4);
    \path [densely dashed] (D3) edge (C1);
    \path [densely dashed] (D3) edge (C2);
    \path [densely dashed] (D3) edge (C3);
    \path [densely dashed] (D3) edge (C4);
    \path [densely dashed] (D4) edge (C1);
    \path [densely dashed] (D4) edge (C2);
    \path [densely dashed] (D4) edge (C3);
    \path [densely dashed] (D4) edge (C4);
\end{scope}
\end{tikzpicture}
}
\caption{\subref{fg:g'_arnborg} the bipartite graph $G'$ transformed from $G$ (Figure \ref{fg:g}) by steps A1-A3; \subref{fg:g_hat_arnborg} the bipartite graph $\hat{G}$ transformed from $G'$ for a given node $c$ by the additional step L4;  \subref{fg:g_hat_wrs_sub_arnborg} the subgraph obtained from $C(\hat{G})$ by removing a simplicial node $r_c^1$ and its excess (specified in Lemma \ref{lm:make_cg'_moral_arnborg}).}

\label{fg:arnborg_reduction}
\end{figure}

\cite{arnborg1987complexity} reduced the MCLA problem to the decision problem of whether or not a graph has a bounded treewidth. Below are the details of \cite{arnborg1987complexity}'s polynomial transformation from a graph $G=(V,E)$ (Figure \ref{fg:g}) that is an instance of the MCLA problem to a bipartite graph $G'=(P\sqcup Q,E')$ (Figure \ref{fg:g'_arnborg}), w.r.t. which $C(G')$ is an instance of the treewidth problem for graphs.
\begin{enumerate}[label=A\arabic*.]
    \item $P = \{u_i \mid i \in \{1, \dots, \Delta(G)+1\}, u \in V\}$,
    \item $Q=\{e_i^j \mid j \in \{1,2\}, e_i \in E\} \cup \{R(u) \mid u \in V\}$, where $R(u) = \{r_u^j \mid j \in \{1, \dots, \Delta(G)+1-d_G(u)\}\}$,
    \item $E'=\{ue_i^j \mid j \in \{1,2\}, e_i \in E \text{ s.t. } u \in V(e_i)\}\cup \{uv \mid v \in R(u), u \in P\}$.
\end{enumerate}
The transformation is similar to that of \cite{yannakakis1981computing} but produces multiple copies for the nodes in $G$. The bipartite graph $G'$ built via the above three steps has no saturated node (unless $G=uv$) for the same reason discussed in the preceding subsection, so $C(G')$ is not moral. To make it moral, the following step 
\begin{enumerate}[label=\^{L}\arabic*.]
     \setcounter{enumi}{3}
    \item for a given node $u \in V$, let $\hat{E} = E' \cup \{S(u^j) \mid j \in [1,\Delta(G)+1]\}$, where $u^j \in P$ is the corresponding node to $u$ and $S(u^j)=\{u^jv \notin E' \mid v \in Q\}$
\end{enumerate}
is applied to each copy $u^j$ of a given node $u \in V$ to make valid of any residual node of $u$ being simplicial in $\hat{G}$ (Figure \ref{fg:g_hat_arnborg}). Although this additional step is applied to all copies, the polynomial complexity is guaranteed by the bounded number $\Delta(G)+1$ of copies of $u$. Again, having a simplicial node is necessary but not sufficient for being moral, so the next lemma proves the morality of $C(\hat{G})$. 

\begin{lemma}
\label{lm:make_cg'_moral_arnborg}
Let $\hat{G}=(P\sqcup Q,\hat{E})$ be the bipartite graph constructed from a graph $G=(V,E)$ by steps A1-A3 \& \^{L}4 for a given node $w \in V$. Then $C(\hat{G})$ is moral. 
\end{lemma}
\begin{proof}
It is easy to see that $r_w^1$ is a simplicial node in $C(\hat{G})$. By removing $r_w^1$ and its excess $\epsilon(r_w^1)=\{S(w^j) \mid j \in [1,\Delta(G)+1]\} \cup \{uv \mid u,v \in Q \text{ s.t. } u,v \neq r_w^1\}$, the resulting subgraph (Figure \ref{fg:g_hat_wrs_sub_arnborg}) is the same as $C_p(G'-r_w^1)$. By Lemma \ref{lm:cp_bip_chordal}, Corollary \ref{cor:chordal_implies_moral} and Theorem \ref{thm:equivalent}, $C(\hat{G})$ is moral. 
\end{proof}

Before proceeding, it is necessary to draw the connection between an ordering w.r.t. which a chain graph $G'$ is defined and the perfect elimination ordering (PEO) of the corresponding triangulated graph $C(G')$.
\begin{lemma}
\label{lm:reverse_chain_ord}
Let $G'=(P\sqcup Q,E')$ be a chain graph w.r.t. an ordering $\alpha$ of $P$ and $\pi_P$ be the reverse of $\alpha$. Then for any ordering $\pi_Q$ of $Q$, the elimination ordering $\{\pi_P,\pi_Q\}$ is perfect for the graph $C(G')$. 
\end{lemma}
\begin{proof}
The neighbour set $N_{C(G')}(\alpha(|P|)) = \{P \setminus \alpha(|P|)\} \cup N_{G'}(\alpha(|P|))$, where each of the two subsets is a clique because of the partition completion. $G'$ is a chain graph implies that for each $i \in [1, |P|-1]$, $N_{G'}(\alpha(|P|)) \subseteq N_{G'}(\alpha(i))$. It follows that each $\alpha(i)$ is adjacent to all nodes in $N_{G'}(\alpha(|P|))$, so $\alpha(|P|)$ is simplicial in $C(G')$. By the same argument, the node $\alpha(|P|-1)$ is simplicial in the subgraph $C(G') - \alpha(|P|)$. Hence, the subgraph of $C(G')$ induced by $P$ can be eliminated recursively according to $\pi_P$. The remaining part is a complete subgraph over $Q$. Hence, any ordering of $Q$ appends to $\pi_P$ forms a PEO of $C(G')$. 
\end{proof}

It has been shown that the steps A1-A3 \& \^L4 and partition completion polynomially transform an instance of the restricted MCLA problem to an instance of the treewidth problem for moral graphs. Based on this transformation, the next lemma proves that a \textit{Yes} answer to the restricted MCLA problem is also a \textit{Yes} answer to the treewidth problem for moral graphs and vice versa. Define the \textit{linear cut value} of $G$ w.r.t. an ordering $\alpha$ as $\max_{1 \le i < |V|} |\{uv \in E \mid \alpha(u) \le i < \alpha(v)\}|$. 
\begin{lemma}
\label{lm:treewdith_nphard}
Given a graph $G=(V,E)$ and a positive integer $k \le |V|$, for any node $v \in V$ the minimum linear cut value w.r.t. an ordering $\alpha$ of $G$ is $k$ with $\alpha^{-1}(v)=|V|$ if and only if the treewidth of the corresponding moral graph $C(\hat{G})$ is $\omega = (\Delta(G)+1)\times (|V|+1)+k$.
\end{lemma}
\begin{proof}
$\hat{G}=(P\sqcup Q,\hat{E})$ is the bipartite graph constructed from $G$ using steps A1-A3 \& \^{L}4 
for a given node $v \in V$. Let $\pi_P$ be an ordering of $P$ s.t. for any node $u_i =\alpha(i) \in V$, the corresponding node $u_i^j \in P$ has order 
\begin{equation}
\label{eq:treewdith_pi_p}
    \pi_P^{-1}(u_i^j)=(\Delta(G)+1)\times i-j+1,    
\end{equation}
where $j \in [1, \Delta(G)+1]$. Furthermore, let $\beta$ be the reverse order of $\pi_P$. Then steps (a) and (b) specify a set $F_{\hat{G}}(\beta)$ of fill in edges w.r.t. $\beta$ to triangulate $C(\hat{G})$, because the sets of neighbours of $P$'s nodes in $\hat{G}+F_{\hat{G}}(\beta)$ form the chain $N(\beta(|P|)) \subseteq \dots \subseteq N(\beta(1))$. By Lemma \ref{lm:reverse_chain_ord}, for any ordering $\pi_Q$ of $Q$, the ordering $\{\pi_P,\pi_Q\}$ is a PEO of the triangulated graph $T=C(\hat{G})+F_{\hat{G}}(\beta)$. For each $i \in [1,(\Delta(G)+1)\times |V|]$, the node $\pi_P(i)$ and its neighbours in the elimination graph $T^{i-1}$ form a clique $K^i$. By going through $\{\pi_P,\pi_Q\}$, it produces a list of cliques that include all maximal cliques in $T$ and consequently the maximum clique. Since each $v^j$ is a saturated node in $P$, all maximal cliques correspond to nodes in $P$ only. 

To calculate the size of each corresponding maximal clique, consider the node $u_i^1$ in the elimination graph w.r.t. $\pi_P$ by removing from $T$ the initial $(\Delta(G)+1)\times (i-1)$ nodes in $P$. The partition completion $C(\hat{G})$ connects $u_i^1$ to $\Delta(G)$ nodes correspond to $u_i$ and $\Delta(G)+1$ nodes correspond to the remaining $|V|-i$ nodes in $V$. In addition, the edge set $\hat{E}$ and the fill in edges $F_{\hat{G}}(\beta)$ connects $u_i^1$ to $\Delta(G)+1-d_G(u_l)$ residual nodes in each $R(u_l)$ for $\sigma(u_l) \ge \alpha^{-1}(u_i)$ and the two edge nodes for each edge $e \in E$ for $\sigma(e) \ge \alpha^{-1}(u_i)$. Let $e = xy \in E$ and assume without loss of generality that $\alpha^{-1}(x) > \alpha^{-1}(y)$. Then define $E_1^i = \{xy \in E \mid \alpha^{-1}(x) \le i < \alpha^{-1}(y)\}$ and $E_2^i = \{xy \in E \mid \alpha^{-1}(y) \le i\}$. The degree of $u_i^j$ in the corresponding elimination graph can be calculated by
\begin{align}
\label{eq:max_clique_linear_cut}
d(u_i^j) = &\Delta(G)+\left[(\Delta(G)+1)\times (|V|-i)\right]+ \nonumber \\
&\left[(\Delta(G)+1) \times i-\sum_{k=1}^i d_G(u_k)\right]+ \nonumber \\
& 2|E_1^i|+2|E_2^i| \nonumber \\ 
=& (\Delta(G)+1)\times (|V|+1)-1+|E_1^i|,
\end{align}
because $\sum_{k=1}^i d_G(u_k)=|E_1^i|+2|E_2^i|$. Note that the reason of having $\Delta(G)+1$ copies of each node in $P$ and two edge nodes for each edge in $G$ is to cancel the terms containing $i$ and $E_2^i$ in the final answer. 

It is obvious that $\max\{E_1^i\mid i \in [1,|V|]\}$ is the linear cut value of $G$. 
If an ordering $\alpha$ gives the \textit{Yes} answer to the restricted MCLA problem of a graph $G$, the treewidth of the corresponding moral graph $C(\hat{G})$ is equal to $\omega$ when triangulating it w.r.t. the ordering $\{\pi_P, \pi_Q\}$, where $\pi_P$ is generated according to $\alpha$ by equation (\ref{eq:treewdith_pi_p}). Conversely, if the treewidth of the moral graph $C(\hat{G})$ is $\omega$ w.r.t. the ordering $\{\pi_P, \pi_Q\}$, the minimum linear cut value of $G$ is $k$ w.r.t. the ordering $\pi_P$ induced over $V$. 
\end{proof}

\begin{theorem}
The treewidth problem for moral graphs is NP-complete. 
\end{theorem}

\begin{proof}
Let $F$ be a set of fill in edges to triangulate a moral graph $G$. It takes polynomial time to find the maximum clique in $G+F$ and test if it is at most $k$, so the treewidth problem is in NP. Hence, the theorem follows from Lemma \ref{lm:treewdith_nphard} and the polynomial transformation from $G$ to $C(\hat{G})$. 
\end{proof}

\subsection{Total states}

\begin{figure}
\centering
\subfigure[]{
\label{fg:g'_wen}
\begin{tikzpicture}[scale=1]
\begin{scope}[every node/.style={circle,draw,fill=black,minimum size=1mm,inner sep=1pt},>={Stealth[black]}]
    \node (A) at (0,0.5) {};
    \node[label=left:$c$] (C) at (-1,-1) {};
    \node (B) at (1,-1) {};
    \node (D) at (-1,-2.5) {};          
    \node[fill=white] (E11) at (0.7,-0.25) {};
    \node[fill=white] (E21) at (-0.7,-0.25) {};
    \node[fill=white] (E31) at (0,-1.2) {};
    \node[fill=white] (E41) at (-1.2,-1.75) {};
 
    \path [-] (A) edge (E11);
    \path [-] (B) edge (E11);
    \path [-] (A) edge (E21);
    \path [-] (C) edge (E21);
    \path [-] (B) edge (E31);
    \path [-] (C) edge (E31);
    \path [-] (D) edge (E41);
    \path [-] (C) edge (E41);
\end{scope}
\end{tikzpicture}
}
\subfigure[]{
\label{fg:g_hat_wen}
\begin{tikzpicture}[scale=1]
\begin{scope}[every node/.style={circle,draw,fill=black,minimum size=1mm,inner sep=1pt},>={Stealth[black]}]
    \node (A) at (0,0.5) {};
    \node[fill=white] (a1) at (-0.2,1) {};
    \node[fill=white] (a2) at (0.2,1) {};    
    \node (C) at (-1,-1) {};
    \node[fill=white,label=left:$r_c^1$] (c1) at (-1.5,-1) {};    
    \node (B) at (1,-1) {};
    \node[fill=white] (b1) at (1.5,-1.2) {};
    \node[fill=white] (b2) at (1.5,-0.8) {};
    \node (D) at (-1,-2.5) {};
    \node[fill=white] (d1) at (-1.3,-3) {};    
    \node[fill=white] (d2) at (-1,-3) {};    
    \node[fill=white] (d3) at (-0.7,-3) {};            
    \node[fill=white] (E11) at (0.7,-0.25) {};
    \node[fill=white] (E21) at (-0.7,-0.25) {};
    \node[fill=white] (E31) at (0,-1.2) {};
    \node[fill=white] (E41) at (-1.2,-1.75) {};
    \path [-] (A) edge (a1);
    \path [-] (A) edge (a2);  
    \path [-] (B) edge (b1);
    \path [-] (B) edge (b2);  
    \path [-] (C) edge (c1);  
    \path [-] (D) edge (d1); 
    \path [-] (D) edge (d2); 
    \path [-] (D) edge (d3);   
    \path [-] (A) edge (E11);
    \path [-] (B) edge (E11);
    \path [-] (A) edge (E21);
    \path [-] (C) edge (E21);
    \path [-] (B) edge (E31);
    \path [-] (C) edge (E31);
    \path [-] (D) edge (E41);
    \path [-] (C) edge (E41);
    
    \path [densely dashed,bend right,blue] (a1) edge (C);
    \path [densely dashed,bend left,blue] (a2) edge (C);   
    \path [densely dashed,blue] (b1) edge (C);
    \path [densely dashed,blue] (b2) edge (C); 
    \path [densely dashed,bend left,blue] (d1) edge (C);
    \path [densely dashed,bend right,blue] (d2) edge (C);
    \path [densely dashed,bend right,blue] (d3) edge (C);
    \path [densely dashed,blue] (E11) edge (C);
\end{scope}
\end{tikzpicture}
}

\subfigure[]{
\label{fg:g_hat_wrs_sub_wen}
\begin{tikzpicture}[scale=1]
\begin{scope}[every node/.style={circle,draw,fill=black,minimum size=1mm,inner sep=1pt},>={Stealth[black]}]
    \node (A) at (0,0.5) {};
    \node[fill=white] (a1) at (-0.2,1) {};
    \node[fill=white] (a2) at (0.2,1) {};    
    \node (C) at (-1,-1) {};    
    \node (B) at (1,-1) {};
    \node[fill=white] (b1) at (1.5,-1.2) {};
    \node[fill=white] (b2) at (1.5,-0.8) {};
    \node (D) at (-1,-2.5) {};
    \node[fill=white] (d1) at (-1.3,-3) {};    
    \node[fill=white] (d2) at (-1,-3) {};    
    \node[fill=white] (d3) at (-0.7,-3) {};            
    \node[fill=white] (E11) at (0.7,-0.25) {};
    \node[fill=white] (E21) at (-0.7,-0.25) {};
    \node[fill=white] (E31) at (0,-1.2) {};
    \node[fill=white] (E41) at (-1.2,-1.75) {};
    \path [densely dashed] (A) edge (B);
    \path [densely dashed] (A) edge (C);
    \path [densely dashed] (A) edge (D);    
    \path [densely dashed] (B) edge (C);
    \path [densely dashed] (B) edge (D);    
    \path [densely dashed] (C) edge (D);
    \path [-] (A) edge (a1);
    \path [-] (A) edge (a2);  
    \path [-] (B) edge (b1);
    \path [-] (B) edge (b2);    
    \path [-] (D) edge (d1); 
    \path [-] (D) edge (d2); 
    \path [-] (D) edge (d3);   
    \path [-] (A) edge (E11);
    \path [-] (B) edge (E11); 
    \path [-] (A) edge (E21);
    \path [-] (C) edge (E21);
    \path [-] (B) edge (E31);
    \path [-] (C) edge (E31);
    \path [-] (D) edge (E41);
    \path [-] (C) edge (E41);
\end{scope}
\end{tikzpicture}
}
\caption{\subref{fg:g'_wen} the bipartite graph $G'$ transformed from $G$ (Figure \ref{fg:g}) by steps W1-W3; \subref{fg:g_hat_wen} the bipartite graph $\hat{G}$ transformed from $G$ by steps W1 \& L2-L4 for a given node $c$;  \subref{fg:g_hat_wrs_sub_wen} the subgraph obtained from $C(\hat{G})$ by removing a simplicial node $r_c^1$ and its excess (specified in Lemma \ref{lm:make_cg'_moral_wen}).}
\label{fg:wen_reduction}
\end{figure}
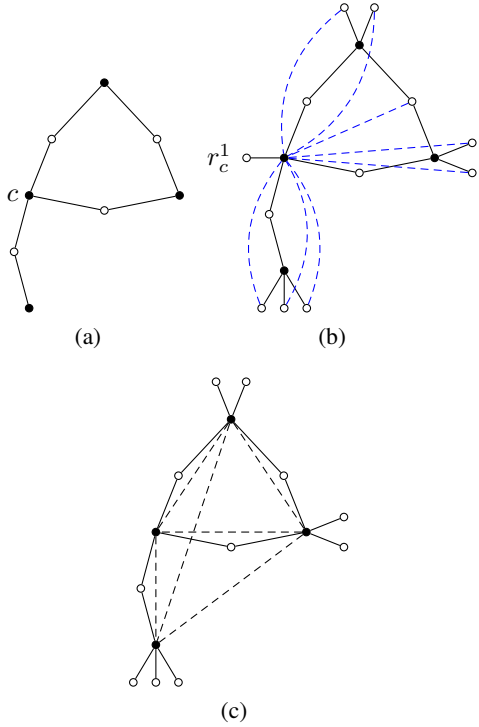
By taking into account the number of states per variable, \cite{wen1990optimal} proposed the total states problem for moral graphs and presented a proof for its difficulty by polynomially reducing the EDS problem to it. His transformation is rather simpler than the previous two cases by making one copy of the nodes in $G$ and one edge node for each edge in $G$, without creating residual nodes (Figure \ref{fg:g'_wen}). The detailed transformation from a graph $G$ (Figure \ref{fg:g}) that is an instance of the EDS problem to a bipartite graph $G'$ is stated in the following steps: 
\begin{enumerate}[label=W\arabic*.]
    \item $P = \{u \mid u \in V\}$,
    \item $Q=\{e_i^1 \mid e_i \in E\}$, 
    \item $E'=\{ue_i^1 \mid e_i \in E \text{ s.t. } u \in V(e_i)\}$. 
\end{enumerate}
It then follows by applying the partition completion on $G'$ to transform it to $C(G')$, which is an instance of the total states problem for graphs but not for moral graphs. The resulting graph $C(G')$, however, encounters the same problem of not satisfying a necessary condition of being moral. Therefore, \cite{wen1990optimal}'s reduction does not prove that the total states problem for moral graphs is NP-complete. 

To prove the EDS problem is reducible to the total states problem for moral graphs in polynomial time, W2 and W3 are replaced by the following two steps 
\begin{enumerate}[label=L\arabic*.]
	\setcounter{enumi}{1}
	\item $Q=\{e_i^1 \mid e_i \in E\} \cup \{R(u) \mid u \in V\}$, where $R(u) = \{r_u^j \mid j \in \{1, \dots, \Delta(G)+1-d_G(u)\}\}$, 
    \item $E'=\{ue_i^1 \mid e_i \in E \text{ s.t. } u \in V(e_i)\}\cup \{uv \mid v \in R(u), u \in P\}$,
\end{enumerate}
to create residual nodes before applying the same step L4 (as in Section \ref{subsec:mini_fillin}) for a given node to get the bipartite graph $\hat{G}$ (Figure \ref{fg:g_hat_wen}).

\begin{lemma}
\label{lm:make_cg'_moral_wen}
Let $\hat{G}=(P\sqcup Q,\hat{E})$ be the bipartite graph constructed from a graph $G=(V,E)$ by steps W1 \& L2-L4. Then $C(\hat{G})$ is moral. 
\end{lemma}
\begin{proof}
$r_c^1$ is the simplicial node in $C(\hat{G})$. The excess removed with it is $\epsilon(r_c^1)=S(c) \cup \{uv \mid u,v \in Q \text{ s.t. } u,v \neq r_c^1\}$. The rest of the proof is the same as Lemma \ref{lm:make_cg'_moral}. 
\end{proof}

For simplicity, define $N(i) = |\{u \mid u\alpha(i) \in E \text{ s.t. } \alpha^{-1}(u) > i\}|$. The following lemma proves the hardness of the total states problem for moral graphs by polynomially reducing it from the restricted EDS problem. 

\begin{lemma}
\label{lm:total_states}
Given a graph $G=(V,E)$ and a sequence of non-negative integers $<d_1,\dots,d_{|V|}>$ not exceeding $|V|-1$, for any node $w \in V$ each value in the sequence satisfies $d_i=N(i)$ w.r.t. an ordering $\alpha$ of $G$ with $\alpha^{-1}(w)=|V|$ if and only if the corresponding moral graph $C(\hat{G})$ has the total number of states equal to $\delta=\sum_{i=1}^{|V|} |V| + \Delta(G) \times i + 1 + \sum_{j=1}^i [d_j - d_G(\alpha(j))]$.
\end{lemma}
\begin{proof}
$\hat{G}=(P\sqcup Q,\hat{E})$ is the bipartite graph constructed from $G$ using the steps W1 \& L2-L4. 
Let $\beta$ be the reverse order of $\alpha$. According to the steps (a) and (b), $\beta$ specifies a set $F_{\hat{G}}(\beta)$ of fill in edges to make $\hat{G}$ a chain graph. Hence, $T=C(\hat{G})+F_{\hat{G}}(\beta)$ is a triangulated graph. By Lemma \ref{lm:reverse_chain_ord}, for any ordering $\alpha_Q$ of $Q$ the ordering $\{\alpha, \alpha_Q\}$ is a PEO of $T$. As stated in the proof of Lemma \ref{lm:treewdith_nphard}, the maximal cliques of $T$ only correspond to nodes in $P$, so the rest of this proof does not consider the degrees of the nodes in $Q$. 

For $i \in [1,|V|]$, $N_{T^{i-1}}(\alpha(i))=N_P\cup N_Q$, where $N_P=\{u \in P \mid \alpha^{-1}(u) > i\}$ and $N_Q=\cup_{j=1}^i N_{T^{j-1}|_Q}(\alpha(j))$ because $\hat{G}+F_{\hat{G}}(\beta)$ is a chain graph. The cardinality of $N_P$ can be easily calculated by $|V|-i+1$. The set $N_Q$ consists of the union of the neighbours of $\alpha(j)$ restricted to $Q$ in the elimination graph $T^{j-1}$ for all $j \in [1,i]$. The restricted neighbour set $N_{T^{j-1}|_Q}(\alpha(j))$ contains $N(j)$ edge nodes incident to $\alpha(j)$ because there is exactly one edge node for each edge in $G$, and $\Delta(G)+1-d_G(\alpha(j))$ residual nodes of $\alpha(j)$. So the total size $|N_Q|=\sum_{j=1}^i N(j)+\Delta(G)+1-d_G(\alpha(j))$. To show $N_P \cup N_Q$ forms a clique, it is easy to see that each of these subsets is a clique because of the partition completion. For any $u \in N_P$, the condition $\alpha^{-1}(u) > i$ implies $\beta^{-1}(u) < i$. Step (a) implies that $\sigma(v) \ge i$ for any $v \in N_Q$. It follows that $\sigma(v) > \beta^{-1}(u)$ for any $u \in N_P$ and $v \in N_Q$, so each node in $N_P$ is connected to each node in $N_Q$ by step (b). Therefore, the closed neighbourhood $N_{T^{i-1}}[\alpha(i)]=N_{T^{i-1}}(\alpha(i)) \cup \{\alpha(i)\}$ is a clique in $T$. It is in fact a maximal clique, because the set $R(\alpha(i))$ is not incident to $\alpha(i-1)$ for $i \in [2, |V|]$. The size $k_i$ of the maximal clique corresponds to $\alpha(i)$ is thus 
\begin{align}
\label{eq:total_states_ki}
    k_i &=|V|-i+1+\sum_{j=1}^i [N(j) + \Delta(G)+1-d_G(\alpha(j))] \nonumber \\
    &= |V| + \Delta(G) \times i + 1 + \sum_{j=1}^i [N(j) - d_G(\alpha(j))].
\end{align}
Assume that all variables considered are binary, the total number of states summing over all $|V|$ maximal cliques is $\sum_{i=1}^{|V|} 2^{k_i}$. 

If there exists an ordering $\alpha$ of $G$, w.r.t. which the EDS answer is \textit{Yes}, substituting $N(i)$ by $d_i$ in equation (\ref{eq:total_states_ki}) entails that the total states of $C(\hat{G})$ is $\delta$ w.r.t. $\alpha$. That is, $\alpha$ is a \textit{Yes} answer to the total states problem for the corresponding moral graph $C(\hat{G})$. Conversely, if the answer to the total states problem for a moral graph $C(\hat{G})$ is \textit{Yes} w.r.t. an ordering $\{\alpha, \alpha_Q\}$, it follows from equation (\ref{eq:total_states_ki}) that $d_i=N(i)$, so $\alpha$ gives a \textit{Yes} answer to the EDS problem for the graph $G$. 
\end{proof}

\begin{theorem}
The total states problem for moral graphs is NP-complete. 
\end{theorem}
\begin{proof}
Since the maximal cliques of a triangulated graph can be found in polynomial time, the total states of any triangulated graph can be verified in polynomial time to be greater than $\delta$ or not. Hence, the problem is in NP. 
Given the moral graph $C(\hat{G})$ can be transformed from a graph $G$ by the steps W1 \& L2-L4 in polynomial time, Lemma \ref{lm:total_states} proves the NP-hardness the total states problem for moral graphs. 
\end{proof}

\section{Conclusion}
Optimal moral graph triangulation plays an important role in determining the computational complexity of the junction tree algorithm for belief propagation on Bayesian networks. The minimum number of fill in edges indirectly but closedly related to the maximum clique size of a triangulated moral graph. The treewidth of a moral graph directly determines the efficiency of the junction tree algorithm when computing probabilities of unobserved variables by marginlizing out observed variables in the largest clique. The total number of states when summing over all maximal cliques in a triangulated moral graph takes into account the number of states per variables. The optimal moral graph triangulation with the objective of minimizing the number of fill in edges, the maximum clique size or the total number of states has proved to be NP-complete in this paper. Therefore, this paper clears the matter that optimal moral graph triangulation with the above constraints was previously proved to be NP-complete. 

\newpage 
\bibliographystyle{abbrvnat}
\bibliography{references}

\end{document}